\documentclass[11pt]{article}
\usepackage{graphicx}    
\usepackage{amssymb,latexsym}
\usepackage{amsmath,amsfonts,amsthm}
\usepackage{comment}
\usepackage{color}
\usepackage{natbib}
\usepackage{float}
\usepackage{booktabs}
\theoremstyle{plain}			

\newtheorem{thm}{Theorem}[section]
\newtheorem{lemma}[thm]{Lemma}
\newtheorem{prop}[thm]{Proposition}
\newtheorem{cor}[thm]{Corollary}

\theoremstyle{definition}		

\input prepictex
\input pictexwd
\input postpictex

\newcommand{\bpic}[4]{\beginpicture
\setcoordinatesystem units  <1pt,1pt>
\setplotarea x from #1 to #2, y from #3 to #4}
\newcommand{\epic}{\endpicture}
\newcommand{\hl}[3]{\put{\line(1,0){#1}} [Bl] at #2 #3 }
\newcommand{\vl}[3]{\put{\line(0,1){#1}} [Bl] at #2 #3 }

\newcommand{\bull}[2]{\put{$\bullet$} at #1 #2 }

\makeatletter
\renewcommand\@biblabel[1]{}

\makeatother

\begin{document}

\title{\Large\bf  Electoral Competition under Best-Worst Voting Rules}

\author{Dodge Cahan\footnote{Department of Economics, University of California, San Diego, USA (dcahan@ucsd.edu).} \quad and \quad   Arkadii Slinko\footnote{Department of Mathematics, University of Auckland, Auckland, New Zealand (a.slinko@auckland.ac.nz).}}

\date{August 2016}

\maketitle
\thispagestyle{empty}


\begin{abstract}
We characterise multi-candidate pure-strategy equilibria in the Hotelling-Downs spatial election  model for the class of best-worst voting rules, in which each voter is endowed with both a positive and a negative vote, i.e., each voter can vote in favour of one candidate and against another one. The weights attached to   positive and  negative votes in calculating a candidate's net score may be different, so that a negative vote and a positive vote need not cancel out exactly. These rules combine the first-place seeking incentives of plurality with the incentives to avoid being ranked last of antiplurality.
We show that these rules generally admit  equilibria, which are nonconvergent if and only if the importance of a positive vote exceeds that of a negative vote. The set of equilibria in the latter case is very similar to that of plurality, except that the platforms are less extreme due to the moderating effect of negative votes. Moreover, any degree of dispersion between plurality, at one extreme, and full convergence, at the other, can be attained for the correct choice of the weights.
\end{abstract}


\section{Introduction}
\label{intro}

Hotelling's \citeyearpar{hotelling} ``Main Street'' model of  spatial competition between firms has---most notably thanks to its adaptation by Downs \citeyearpar{downs} to ideological competition of political parties---enjoyed a significant presence in political science literature. 
In the classical model, there is a society of voters whose ideal policy platforms lie along the left-right political spectrum. A set of exogenously given political candidates or parties choose  platforms to advocate so as to maximise their support from the voters, who vote for the candidate with the platform nearest to his or her personal ideal platform.  

Most such studies of Downsian competition have focused on situations in which elections are held under the voting system known as plurality rule. This is the simplest system where voters have one vote each, which they cast for their favourite candidate, and whoever gets the most votes wins. Under plurality, voters' second, third and other preferences---and most importantly for this paper, their last preferences---do not matter.
However, voting systems, both  used in practice and  studied theoretically, are diverse. Many of them do take into account voters' partial or full ranking of candidates when producing a winner. These include, among others, approval voting, Borda count, and single transferable vote. When the preferences beyond first matter, candidates' incentives change, and we expect equilibrium outcomes to vary as well. In this paper, we analyze the equilibrium properties of a largely overlooked class of voting rules, which combine positive and negative voting, and are referred to as \textit{best-worst}  rules~\citep{garciamartinez}. Under these rules, voters  have both one positive and one negative vote, whose weights (in absolute value) are allowed to be different, while their intermediate preferences are inconsequential. Plurality, anti-plurality, and the system in which voters have one positive and one negative vote of equal absolute value,\footnote{Also known as ``single-positive-and-single-negative'' voting  \citep{myerson1}.} are special cases.

 The main result of this paper is that in the Hotelling-Downs model there is a close link between the  pure-strategy equilibria of general best-worst rules and those of plurality, which is well known to admit divergent equilibria in which candidates adopt a range of ideologically diverse positions \citep{eatonlipsey, denzaukatsslutsky}.   When the importance of a positive vote exceeds that of a negative vote, equilibria take the same general form as those of plurality, with divergent policy platforms   advocated. However, the key difference is that,  while differentiated, the equilibrium platforms for the best-worst rules are more moderate and less spread out. Indeed, these rules present candidates with a clear centrifugal motive to seek first-place rankings, as occurs under plurality, but with the   simultaneous  incentive to avoid being the most unpopular candidate and receiving negative votes. This last property encourages a degree of policy moderation---adopting extreme platforms is discouraged as doing so is likely to single oneself out as a target for the negative votes of voters at the opposite end of the ideological spectrum. As the importance of a negative vote increases relative to that of a positive vote, the equilibrium platforms move inwards towards the median voter's ideal platform. Eventually all platforms merge at the median as a negative vote reaches parity with  a positive vote (i.e., one negative vote cancels out one positive vote exactly). When  a negative vote becomes more important than a positive vote, only convergent equilibria exist, with no policy differentiation.

 Describing the equilibrium properties of different voting systems is an important task for several reasons. When choosing between voting rules, first of all, we would like to know whether or not equilibria exist---their absence may lead to permanent instability and a  lack of predictability of outcomes \citep{myerson1}. Second, if they exist, an electoral designer would prefer a rule that admits equilibria with desirable properties. 
 The  main consideration  here is a tradeoff between discouraging extremism and promoting fair representation---it is undesirable if extremist platforms have a better chance of attracting votes than more centrist platforms, while at the same time the rational for voting in the first place is to provide citizens with the ability to choose  parties or politicians that represent their varied interests. 
 
 Our results show that best-worst rules do well on both counts. They admit nonconvergent equilibria, offering voters with a choice over different platforms and avoiding Hotelling's ``excessive sameness''. At the same time, the perhaps excessive extremism associated with plurality \citep{cox1,cox85,myerson4,lasliermaniquet}   is moderated. Indeed, depending on the weight placed on a negative vote, we may have any level of dispersion of platforms between that of plurality, at one extreme, and full convergence of positions, at the other.  Moreover, they are simple and easily implementable.

The idea of voting against candidates has been around for some time. \cite{boehm} in an unpublished essay suggested that voters in an election be allowed either to cast a vote for or against a candidate, but not both. A candidate's ``negative'' votes would be subtracted from his ``positive'' votes to determine his net vote, and the candidate with the highest net vote would win.\footnote{Presumably, if the highest net vote is negative, then nobody is elected.}
\citeauthor{boehm}---and many others after him (see, e.g., \citealp{leef})---argued that the introduction of negative votes in United States presidential elections would force the candidates to appeal to  voters with positive programs, rather than just filling  the airwaves with ads attacking other candidates, sowing doubt among their supporters. The rule suggested by \citeauthor{boehm} is now known as negative voting \citep{brams}.   Anti-plurality voting is a similar method in which each voter votes against a single candidate, and the candidate with the fewest votes against wins. In other words, anti-plurality determines who among the candidates is the least unpopular.\footnote{Anti-plurality is also sometimes referred  to as negative voting as well as ``veto'' \citep{kang}.}

The use of some form of negative voting in elections is not so uncommon. For example, Nevada gives voters the option to vote against all candidates by having a ``None of these candidates'' option on the ballot. Prior to 2000, Lithuanian voters were allowed voters to express approval, neutrality or disapproval of candidates in the proportional representation part of their parliamentary elections \citep{renwickpilet}. The Senate of Spain uses a rule related to anti-plurality \citep{grofmanlijphart}---provinces elect four senators, and voters may vote for any three.\footnote{More accurately, this is closer to $k$-plurality, where voters have $k$ votes to cast, since the number of candidates could exceed four. Also, the Spanish method is more complicated in that it allows for abstention.}

 The possibility of casting both positive and negative votes, as occurs under best-worst rules, may be even more advantageous. First, the number of ways in which voters can express themselves is further diversified, which could increase  turnout  by appealing to those voters who are   enticed more by the ability to vote against a candidate than for one~\citep{kang,leef}. Second, it can give a fighting chance to major or minor centrist parties---it is not unthinkable that people on the extreme left will vote for a leftist candidate and against a right-wing one, while the right-wing voters will do the opposite. Their votes will cancel out and a centrist candidate will be elected.\footnote{See \cite{kang} for an account of more costs, benefits and tradeoffs associated with negative voting.}

The rest of this paper is organised as follows: in Section~\ref{rellit} we outline some literature related to this work; in Section \ref{themodel} we present the model; in Section \ref{results} we present our main results; Section \ref{conclusion} provides our concluding remarks. A few minor and auxiliary results are presented in the Appendix.

\section{Related literature}
\label{rellit}

Best-worst rules specifically and notions related to them have been considered in the literature before in other contexts. 
The idea that the best and worst alternatives play a special role in the decision process has been prominent in decision theory. For example, the Arrow-Hurwicz     \citeyearpar{arrowhurwicz} criterion for choice under uncertainty takes a weighted average of the best and worst expected value/utility outcomes and does not take into account intermediate outcomes, and \citet{marleylouviere}  look at probabilistic discrete choice models through the best-worst lense.

 \citet{garciamartinez} provide an axiomatic characterisation of the class of best-worst voting rules considered in this paper. \citet{alcantudlaruelle} characterise a related voting rule in which, for each candidate, voters may express approval, indifference, or disapproval. This rule is also studied in \citet{felsenthal} from the perspective of voter strategies. \citet{joymcmunigal} believe that the current system of peremptory challenges in the criminal justice system of the United States makes it easy to exclude qualified African Americans jurors in the process of jury selection and propose that it be replaced with a system of peremptory strikes and peremptory inclusions.  In other words, both the defense and the prosecution should be allowed to not just rule  potential jurors out, but also ``rule them in''. 

None of these papers, however, look at how the incentives created by these voting systems affect political competition.
Given the  very natural combination of  negative and positive voting embodied in the best-worst rules, it is surprising that, to the best of our knowledge, they have been overlooked in the spatial competition literature. There are, however, a few results related to some special cases of best-worst rules. Most notably, plurality has been extensively discussed, and its equilibrium properties are characterised in \citet{eatonlipsey} and \citet{denzaukatsslutsky}. Anti-plurality is known to allow convergent equilibria  in which all candidates adopt the same policy platform, but not to allow nonconvergent equilibria \citep{cox1}. 

The two most relevant papers to this research are~\cite{cox1} and~\cite{cahanslinko}. Both  are concerned with Nash equilibria under the class of voting rules known as general scoring rules, of which the best-worst rules are a subclass.   \citet{cox1} characterised all scoring rules that have 
convergent Nash equilibria, which leads to  a straightforward description of all best-worst rules allowing convergent equilibria, as we will describe in Section~\ref{results}.  However, Cox's theorem says nothing about the possibility of divergent equilibria, which is the focus of the paper \cite{cahanslinko} and also this paper. 

\citet{cahanslinko} investigate the existence and properties of nonconvergent equilibria under general scoring rules. In some subclasses of scoring rules---in particular, those whose score vector is convex---they managed to characterise all rules that allow Nash equilibria. These rules appear to be truncated variants of the Borda rule. This result is, however, inapplicable to the best-worst rules considered in this paper since their score vector is neither convex nor concave. 
A general characterisation of  scoring rules that allow equilibria remains an open question.


\section{The model}
\label{themodel}

There is a unit mass of consumers distributed uniformly on the interval [0,1], the issue space.\footnote{We do not actually need the distribution to be literally uniform. All we need is that the candidates believe this to be the case or take it as a simplifying assumption in their decision-making \citep{aragonesxefteris,cahanslinko}.}  There are $m$ candidates---candidate $i$'s position is $x_i$, and a strategy profile $x=(x_1,\ldots,x_m)\in [0,1]^m$ describes the platforms of all the candidates. A strategy profile implies a set of distinct occupied positions, $x^1<x^2<\ldots<x^q$. We denote by $n_i$ the number of candidates at occupied position $x^i$ and we will sometimes use the alternative notation for a strategy profile, $x=((x^1,n_1),\ldots,(x^q,n_q))$, which gives the location and number of candidates at each occupied position rather than each individual candidate's position. 

We will use notation $[n]=\{1,\ldots,n\}$ and if $I=[a,b]$ is an interval, then $\ell(I)=b-a$ is the length of the interval. We assume sincere voters with single-peaked, symmetric utility functions who, hence, rank candidates according to the distance between their advocated platform and the voter's ideal position. Voters who are indifferent between candidates decide on a strict ranking by  fair lottery.

A best-worst voting rule can be described as follows: a first-place ranking earns a candidate a normalised 1 point, while a last-place ranking earns the candidate $-c$ points, where $c\geq 0$. Being ranked anywhere other than first or last by a voter earns a candidate nothing. The magnitude of $c$ describes the relative importance of the positive vote relative to the negative vote, which is the parameter of interest here. Thus, a rule can be described by a pair of numbers $s=(c,m)$, where $m$ is the number of candidates.\footnote{As noted in the Introduction, best-worst rules belong to the class of general scoring rules. A scoring rule is a vector $s=(s_1,\ldots,s_m)$, where $s_1\geq \cdots \geq s_m$, $s_1>s_m$, and $s_i$ is the number of points assigned to the $i$-th ranked candidate in a voter's ballot. A best-worst rule $s=(c,m)$, then, is equivalent to scoring rule $s=(1,0,\ldots,0,-c)$.}

Candidate $i$'s score is the total number of weighted positive and negative votes received, denoted $v_i(x)$. Candidates choose positions simultaneously so as to maximise $v_i(x)$.\footnote{See Stigler's \citeyearpar{stigler} argumentation for this assumption and a discussion of it in \citet{denzaukatsslutsky}.} Our equilibrium concept is the Nash equilibrium in pure strategies. Profile $x^*=(x_1^*,\ldots,x_m^*)$ is an equilibrium if and only if  $v_i(x^*)\geq v_i(x_i,x_{-i}^*)$ for all $i\in [m]$ and for all $x_i\in [0,1]$, where $(x_i,x_{-i}^*)=(x_1^*,\ldots,x_{i-1}^*, x_i,x_{i+1}^*,\ldots,x_m^*)$. A convergent Nash equilibrium (CNE) is an equilibrium in which all candidates adopt the same platform, while in a non-convergent Nash equilibrium (NCNE), at least two of the platforms are distinct.

\section{Results}
\label{results}

Our main result is  a general characterisation of NCNE for rules $s=(c,m)$ in Theorem \ref{abb0general}.  Before we concentrate on NCNE, however, we should address the issue of CNE---equilibria in which all candidates adopt the same platform. In fact, their characterisation is straightforward, presented below in Proposition \ref{CNE}. This result is due to Cox  \citeyearpar{cox1}, who  characterised CNE for general scoring rules, a broad class of voting rules to which best-worst rules belong.  

\begin{prop}[\citealp{cox1}]\label{CNE} A rule $s=(c,m)$ admits CNE if and only if $c\geq 1$, in which case the profile $x=((x^1,m))$ is a CNE for any $x^1\in \left[\frac{m-1+c}{m(1+c)}, 1-\frac{m-1+c}{m(1+c)}\right]$.\end{prop}

\begin{proof} 
 For $x=((x^1,m))$ to be a CNE, it should not be beneficial to deviate just to the left or right of $x^1$. That is, we have CNE if and only if: first,
 $v_i(x^{1-},x_{-i})=x^1-c(1-x^1)\leq \frac{1-c}{m}=v_i(x)$; and, second, $v_i(x^{1+},x_{-i})=1-x^1-cx^1\leq  v_i(x)$. Together, these two conditions  yield  the interval of possible CNE, which is nonempty if and only if $c\geq 1$.
\end{proof}

Proposition \ref{CNE} tells us that CNE can only exist if the weight on a positive vote does not exceed that of a negative vote. In this case, deviating differentiates a candidate in a positive way for one side of the electorate, and negatively for the other. The   gain in terms of positive votes is not worth the damage due to the negative votes that the candidate will now receive, so candidates will stay put at the CNE platform. If $c\geq 1$, then CNE exist at any point of an interval centered at the median voter's ideal position. As $c$ increases, this interval expands, meaning that a wider range of CNE are possible.\footnote{Provided $m>2$; if $m=2$, any rule reduces to plurality.}

While Proposition \ref{CNE} tells us everything there is to know about CNE, it is silent about NCNE. We do know that NCNE exist for plurality \citep{eatonlipsey}, but not for antiplurality \citep{cox1}, both of which are examples of best-worst rules, so the picture is not at all clear in general. 

It turns out that, for NCNE to exist, it must be that $c<1$. In other words,  the value of a positive vote must outweigh the value of a negative vote in order for the candidates to be induced to adopt divergent policies, otherwise the centripetal incentive to avoid being singled out as the worst candidate is too strong and only CNE can exist. This also implies that CNE and NCNE cannot exist simultaneously for the same rule.\footnote{They can coexist for other scoring rules \citep{cahanslinko} outside the class of best-worst rules.}

\begin{prop}\label{BRonly}  The rule $s=(c,m)$ does not admit NCNE if $c\geq 1$.
\end{prop}

\begin{proof} Consider candidate 1 at position $x^1$, which is occupied by $n_1$ candidates, where $2\leq n_1\leq m-2$. Consider intervals $I_1=[0,x^1]$ and $I_2=[(x^1+x^q)/2,1]$.   If 1 makes an infinitesimal move to the right of $x^1$, then in the rankings of voters in $I_1$, of which there is positive measure by Lemma \ref{noextremepositions}, she falls behind the other $n_1-1$ candidates originally at $x^1$, thus losing their positive votes. On the other hand, 1 rises ahead of these $n_1-1$ candidates in the rankings of all other voters and, in particular, no longer receives a negative vote from any voter. Then, the score candidate 1 loses by making this move  is $v_{lost}=\frac{1}{n_1}\ell(I_1)$. On the other hand, 1's gain from this move is $v_{gained}=\frac{1}{n_1}c\ell(I_2).$

For NCNE, it must be the case that $v_{lost}\geq v_{gained}$, or $\ell(I_1)\geq c\ell(I_2)$. Since we assume $c\geq 1$, this implies that $\ell(I_1)\geq \ell(I_2)$, or $x^1\geq 1-(x^1+x^q)/2$.
  Similar considerations with respect to candidate $q$ yields the requirement that $\ell([x^q,1])\geq \ell([0,(x^1+x^q)/2]$, or $1-x^q\geq (x^1+x^q)/2$. Together, these two conditions imply that $x^1\geq x^q$, an impossibility for an NCNE.
\end{proof}

Before we proceed to our characterisation,   some additional notation. Let
\begin{enumerate}
\item[(i)] $I_1=[0,(x^1+x^2)/2]$,
\item[(ii)] $I_i=[(x^{i-1}+x^{i})/2,(x^{i}+x^{i+1})/2]$ for $2\leq i \leq q-1$,
\item[(iii)]  $I_q=[(x^{q-1}+x^{q})/2,1]$,
\end{enumerate}
 be the ``full-electorates" corresponding each occupied position. A full-electorate $I_i$ is the set of voters for whom a given occupied position $x^i$  is the nearest, so that any candidates located there are ranked first equal for these voters. For each $i\in [q]$ let $I_i^L=\{y\in I_i: y\leq x^i\}$ and $I_i^R=\{y\in I_i: y\geq x^i\}$ be the left and right ``half-electorates" whose union is the full-electorate $I_i$. That is, we simply partition a full-electorate into those voters whose ideal positions lie to the left of the given occupied position and those who lie to the right. We note that  $\ell(I_i^R)=\ell(I_{i+1}^L)$ for $i\in [q-1]$.

We now present our characterisation of NCNE for best-worst rules, Theorem \ref{abb0general}, which provides five necessary and sufficient conditions for a profile to be an NCNE for a given best-worst rule. Condition (i) states that the outermost occupied positions must be occupied by two candidates apiece. It is clear that they cannot be single candidates, but this condition also excludes the possibility of more than two candidates, as in the well-known case of plurality (see \citealp{eatonlipsey}). The second condition says that all paired candidates' half-electorates are the same length, excluding end electorates, while (iii) relates these interior half-electorates to the outermost half-electorates. Conditions (iv) and (v) put restrictions on the lengths of various electorates: first, an unpaired candidate's full-electorate cannot be smaller than any half-electorate (excluding end half-electorates); and, second, a paired candidate's half-electorate cannot be smaller than an unpaired candidate's half-electorate (excluding end half-electorates).  

An important observation to make is that, with the exception of (iii), all the remaining conditions  are identical for any rule---they do not depend directly on $c$, as long as  $c<1$. Condition (iii) demonstrates that the equilibrium spacing of the candidates will be affected by $c$, but not the configuration of candidates (i.e., the number of occupied positions and how many candidates occupy them). Thus, if they exist (we will see shortly that they do for any rule with $c<1$), NCNE for best-worse rules will have the same general form as NCNE for plurality, the only difference being the exact location of $x^1,\ldots,x^q$.
 
\begin{thm}\label{abb0general}
Given a   rule $s=(c,m)$, with $c<1$, the following conditions are necessary and sufficient for a profile $x$ to be an NCNE:
\begin{itemize}
\item[(i)]$n_i\leq 2$ for all $i\in [q]$ and $n_1=n_q=2$. That is, candidates at the most extreme occupied positions are paired.
\item[(ii)] If $n_i=2$ for $1<i<q$, then $\ell(I_i^L)=\ell(I_i^R)=\ell(I_1^R)=\ell(I_q^L)$. Let $I^p$ denote this common measure. That is, all paired candidates' half-electorates are the same length (except end half-electorates).
\item[(iii)] $\ell(I_1^L)=\ell(I_q^R)=I^p+\frac{c}{2}$.
\item[(iv)] If $n_i=1$, then  both $ \ell(I_i)\geq \ell(I_k^L)$ for all $k\ne 1$ and $ \ell(I_i)\geq\ell(I_k^R)$ for all $k\ne q$.  That is, any (unpaired) candidate's full-electorate is no smaller than any other half-electorate (excluding end half-electorates).
\item[(v)] $I^p\geq \ell(I_k^L)$ for all $k\ne 1$ and $I^p\geq \ell(I_k^R)$ for all $k\ne q$. That is, a paired candidate's half-electorate (excluding end half-electorates) is no smaller than any other (unpaired) candidate's half-electorate (excluding the end half-electorates).
\end{itemize}

\end{thm}

\begin{proof} That (i) is necessary follows from  Lemma \ref{nomorethantwo}, so we start by showing the necessity of (ii).
Suppose candidate  $j$ is at $x^i$, where $n_i=2$ and suppose without loss of generality that $\ell(I_i^L)>\ell(I_i^R)$. Then 
$ v_j(x^{i-},x_{-j})= \ell(I_i^L) > \ell(I_i^R) =v_j(x^{i+},x_{-j}), $ contradicting Lemma \ref{45cand2}. So $\ell(I_i^L)=\ell(I_i^R)$. Let $I^p$ denote this common measure.
Moreover, note that $v_1(x^{1+},x_{-1})=\ell(I_1^R)$. Using Lemmas \ref{45cand2} and \ref{equalscores}, we know $v_1(x^{1+},x_{-i})=v_1(x)=v_j(x)=I^p$, so that  $\ell(I_1^R)=I^p$. Similarly, $\ell(I_q^L)=I^p$. Hence, condition (ii) is necessary.

Now condition (iii). Note that we must have $\ell(I_1^L)=\ell(I_q^R)$. Otherwise,  if $\ell(I_1^L)>\ell(I_q^R)$, then, using Lemmas \ref{45cand2} and  \ref{equalscores},$$v_q(x)=v_1(x)=v_1(x^{1-},x_{-1})>v_q(x^{q+},x_{-q})=v_q(x),$$ a contradiction. Thus, $(x^1+x^q)/2=1/2$. Hence, 
$ v_1(x)=\frac{1}{2}( \ell(I_1)+I^p)-\frac{c}{4}$,
which, by Lemma \ref{45cand2}, is equal to $v_1(x^{1+},x_{-1})= I^p  $,
so $\ell(I_1^L)=I^p+\frac{c}{2}$.

Now conditions (iv) and (v). Let candidate $l$ be at $x^i$. Then,  if $n_i=1$, we have $v_l(x)= \ell(I_i) $. Suppose there is some $k>1$ such that $\ell(I_i)<\ell(I_k^L)$. Clearly the half electorate $I_k^L$ could not be $i$'s half electorate, i.e. $k=i$ or $k=i+1$. So we have
$$v_l(x^{k-},x_{-l})= \ell(I_k^L) > \ell(I_i) =v_l(x),$$
so this is not an NCNE. So we must have $\ell(I_k^L)\leq \ell(I_i)$. For (v), if $n_i=2$, then (noting that $i$ can be 1 or $q$ since all paired candidates receive the same score by Lemma \ref{equalscores}) $v_l(x)= I^p $, which, to avoid contradiction, implies $I^p\geq \ell(I_k^L)$ for all $k\ne 1$. Similarly for right electorates.

Now sufficiency. We need to check that no candidate can deviate profitably.
Consider candidate $i$ at $x^j$, where $n_j=2$ ($i$ could be an end candidate). We know  that all paired candidates get the same score, $v_i(x)=I^p$, and that $v_i(x^{1-},x_{-1})=v_i(x)$, so $i$ would not want to deviate to $x^{1-}$ or $x^{q+}$. Also, if $t\in (x^k,x^{k+1})$ for some $k<q$, then $v_i(t,x_{-1})=\ell(I_k^R)\leq I^p =v_i(x)$ by condition (v). 
Candidate $i$ would also not deviate to an occupied position $x^k$, $k\ne j$. Doing so would yield a score of $v_i(x^k,x_{-i})=\frac{2}{3}I^p<v_i(x)$ if $n_k=2$ or a score of $v_i(x^k,x_{-i})=\frac{1}{2}\ell(I_k)=\frac{1}{2}(\ell(I_k^L)+\ell(I_k^R))\leq I^p=v_i(x)$ if $n_k=1$, by (v).  So no paired candidates would deviate.

Consider an unpaired candidate $i$ at position $x^j$. Then $v_i(x)= \ell(I_j) $. Clearly any moves within the interval $(x^{j-1},x^{j+1})$ do not change $i$'s score. Suppose $t\in (x^k,x^{k+1})$ for some $k\notin \{j-1,j,q\}$. Then $v_i(t,x_{-i})= \ell(I_k^R)\leq\ell(I_j)=v_i(x)$, so $i$ will not move to any unoccupied position. Suppose $n_k=2$ and $k\notin \{j-1,j+1\}$. Then $v_i(x^k, x_{-i})=\frac{2}{3} I^p < I^p \leq  \ell(I_j) =v_i(x)$, by (iv). Suppose $n_k=1$ and $k\notin \{j-1,j+1\}$. Then $v_i(x^k,x_{-i})=\frac{1}{2} (\ell(I_k^L)+\ell(I_k^R)) \leq  \ell(I_j) =v_i(x)$. So no unpaired candidate wants to deviate to any occupied position that is not adjacent to the candidate's current position.

Finally, we check that no unpaired candidate would move to an adjacent occupied position. If $n_{j-1}=2$, $j-1\ne 1$, then $v_i(x^{j-1},x_{-i})=\frac{1}{3} (I^p+\ell(I_j)) \leq \frac{2}{3} \ell(I_j) <\ell(I_j)=v_i(x)$. If $j-1=1$ then $v_i(x^{j-1},x_{-i})=\frac{1}{3}(\ell(I_1^L)+\ell(I_j))-\frac{c}{6}=\frac{1}{3} \left( I^p+\frac{c}{2}+\ell(I_j)\right)-\frac{c}{6}\leq \frac{2}{3}\ell(I_j)<\ell(I_j)=v_i(x)$. If $n_{j-1}=1$, then $v_i(x^{j-1},x_{-i})=\frac{1}{2}(\ell(I_{j-1}^L)+\ell(I_j))\leq \ell(I_j)=v_i(x)$. So no unpaired candidate wants to move to the next left occupied position or, by similar arguments, to the next right occupied position.
We have checked all possible deviations, so $x$ is a NCNE.
 \end{proof}

While Theorem \ref{abb0general} gives necessary and sufficient conditions for an NCNE, it is not yet clear that these conditions can be satisfied for an arbitrary number of candidates and any $c<1$. The next result addresses this question and shows that they do indeed  exist for any $m\ne 3$. 

\begin{cor}\label{NCNEexist} For all $m\geq 2$, except for $m=3$, NCNE exist for rules $s=(c,m)$ with $c<1$, and for $m\geq 6$ there are infinitely many NCNE for a given rule. Moreover, all NCNE  take the same general form as plurality (they have NCNE with the same number of occupied positions, $q$, with the same number of candidates, $n_i$, at each one, but perhaps  different locations).\end{cor}

\begin{proof} For $m=3$, nonconvergent equilibria cannot exist by the familiar argument \citep{eatonlipsey} that one of them would have to be alone at an outermost occupied position, and would have an incentive to move inwards.  

Consider $m\geq 4$. Suppose candidates are positioned so that all half-electorates have the same length, except for end electorates, and $n_1=n_q=2$. That is, $\ell(I_k^L)=\ell(I_j^R)=I^p$ for all $k\ne 1$, $j\ne q$. 
Then, we place $x^1$ and $x^q$ so that (iii) is satisfied, from which it follows that $I^p=\frac{1}{2q}(1-c)$. By construction, then,  (iv) and (v) are satisfied, and we have an NCNE.

Next, we show that there are infinitely many NCNE for $m\geq 6$. If $m$ is even, construct a profile as above, but with $q=(m+2)/2$ occupied positions,  all of them occupied by two candidates except for the two innermost positions, $x^k$ and $x^{k+1}$, which are occupied by only one candidate each, and all half-electorates except for the outermost of the same length. This will be an equilibrium by the argument of the previous paragraph, with $x^1$ and $x^q$ chosen to satisfy (iii). Let us increase the length of each half-electorate except for $I_k^R$ and $I_{k+1}^L$ by $\epsilon>0$, so that $I^{p'}=I^p+\epsilon$ (that is, we are moving all positions inwards at the expense of the interior two candidates). This maintains (i)-(iii). Condition (iv) will still be satisfied since the only unpaired candidates are those at $x^k$ and $x^{k+1}$, who have full-electorates of length $\ell(I_k)=I^{p'}+\ell(I_k^R)>\max\{I^{p'}, \ell(I_k^R)\}$. Clearly, (v) will still be satisfied, since we are increasing the length of $I^p$ and decreasing the length $I_k^R$ and $I_{k+1}^L$.

If there is an odd number of candidates $m\geq 7$ we can do a similar thing. We start with $q=(m+3)/2$ occupied positions, symmetric about the median, which is occupied by a single candidate. The two innermost occupied positions to the left and right of the median are also occupied by single candidates. Label the   occupied position  at the median  as  $x^k$. All  occupied positions other than these three have two candidates apiece, and we place them so that all half-electorates except the outermost are of the same length. Again, choose $x^1$ and $x^q$ so that (iii) is satisfied. Now, increase the length of all half-electorates except for $I_k^L$ and $I_k^R$ by $\epsilon>0$. As above, this maintains (i)-(iii). Condition (iv) will clearly still be satisfied for all $i\ne k$. For $I_k$, the full electorate is getting smaller, but $\ell(I_k)=2I^p-J\epsilon>I^p+\epsilon$  for small $\epsilon$, where $J$ is the number of half intervals to one side of the median that increase in length. Condition (v) will still be satisfied, since the paired candidates' half-electorates $I^p$ are increasing in length, while the unpaired candidates' half-electorates are either increasing at the same rate, or getting smaller in the case of $I_k^L$ and $I_k^R$.
\end{proof}

 An important consequence of Theorem \ref{abb0general} is that plurality rule produces the most dispersed equilibria, while incorporating a negative vote pulls the platforms inward.  Essentially, the correct choice of $c$ allows an election designer to pick any level of dispersion between that of plurality and full convergence, an important result given the tradeoff between moderation and representation discussed in the Introduction. This is stated in Corollary \ref{extremeplurality}.

\begin{cor}\label{extremeplurality} For a given configuration of candidates, i.e., fixing $q$ and $n_i$ for $i\in [q]$, but allowing $x^i$ to vary, the most extreme equilibria occur under plurality.\end{cor}

\begin{proof} 
Given a number of occupied positions $q$ and the number of candidates $n_i$ at each of them, maximising dispersion consists, essentially, in minimising the location of $x^1$. By (iii), then, we want to minimise $I^p$. Looking at condition (v), we can see that we will want to to have $\ell(I_k^L)=\ell(I_k^R)=I^p$, which will then imply that (iv) is satisfied. This will partition the issue space into $2q$ intervals $I^p$ and two intervals of length $c/2$. Thus, $2qI^p+c=1$, so that $I^p=\frac{1}{2q}(1-c)$ and, hence, $x^1=\frac{1}{2q}(1+c(q-1))$. Increasing $c$, then, increases $x^1$, leading to less dispersed equilibria.
\end{proof} 

In the case of four or five candidates, there is a unique NCNE. 

\begin{cor}\label{45candidates} If $m=4$, then there is a unique NCNE, given by profile $x=((x^1,2),(1-x^1,2))$, where 
$x^1=\frac{1}{4}(1+c).$
If $m=5$, then there is a unique NCNE, given by profile $x=((x^1,2), (1/2,1),(1-x^1,2))$, where
$x^1=\frac{1}{6}(1+2c).$
\end{cor}

In the four- and five-candidate cases,  as expected by Corollary \ref{extremeplurality}, the amount of dispersion observed in the candidates' positions depends on the value of $c$ and is maximal when the rule is plurality. 
As $c$ grows towards $1$, the positions converge at the median voter position. As $c$ increases beyond this point, by Cox's \citeyearpar{cox1} characterisation of CNE,  we know that infinitely  many CNE are possible in an interval that becomes increasingly wide. Hence, there is a
bifurcation point that divides CNE from NCNE when $c=1$. As we move away from this point, more extreme positions are possible---on one side they take the form of CNE, and on the other side they are NCNE.

The six candidate case admits infinitely many equilibria  for a given rule, but the pattern is similar.

\begin{figure}[H]
\bpic{0}{300}{-5}{200} 
\put{$c=0$} at -40 0
\hl{320}{0}{0}
\vl{10}{0}{-5}
\vl{10}{320}{-5}
\bull{40}{-4}
\bull{40}{4}
\bull{120}{0}
\bull{200}{0}
\bull{280}{-4}
\bull{280}{4}
\put{$c=0.25$} at -40 40
\hl{320}{0}{40}
\vl{10}{0}{35}
\vl{10}{320}{35}
\bull{70}{36}
\bull{70}{44}
\bull{130}{40}
\bull{190}{40}
\bull{250}{36}
\bull{250}{44}
\put{$c=0.5$} at -40 80
\hl{320}{0}{80}
\vl{10}{0}{75}
\vl{10}{320}{75}
\bull{100}{76}
\bull{100}{84}
\bull{140}{80}
\bull{180}{80}
\bull{220}{76}
\bull{220}{84}
\put{$c=0.75$} at -40 120
\hl{320}{0}{120}
\vl{10}{0}{115}
\vl{10}{320}{115}
\bull{130}{124}
\bull{130}{116}
\bull{150}{120}
\bull{170}{120}
\bull{190}{116}
\bull{190}{124}
\put{$c=1$} at -40 160
\hl{320}{0}{160}
\vl{10}{0}{155}
\vl{10}{320}{155}
\bull{160}{175}
\bull{160}{169}
\bull{160}{163}
\bull{160}{157}
\bull{160}{151}
\bull{160}{145}
\epic 
\caption{Maximally dispersed NCNE for different choices of $c$.}\label{mostdispersed}

 \vspace{.3cm}

\bpic{0}{300}{-5}{120}
\put{$c=0$} at -40 0
\hl{320}{0}{0}
\vl{10}{0}{-5}
\vl{10}{320}{-5}
\bull{53}{-4}
\bull{53}{4}
\bull{160}{-4}
\bull{160}{4}
\bull{266}{-4}
\bull{266}{4}
\put{$c=0.25$} at -40 40
\hl{320}{0}{40}
\vl{10}{0}{35}
\vl{10}{320}{35}
\bull{80}{36}
\bull{80}{44}
\bull{160}{36}
\bull{160}{44}
\bull{240}{36}
\bull{240}{44}
\put{$c=0.5$} at -40 80
\hl{320}{0}{80}
\vl{10}{0}{75}
\vl{10}{320}{75}
\bull{106}{76}
\bull{106}{84}
\bull{160}{76}
\bull{160}{84}
\bull{212}{76}
\bull{212}{84}
\put{$c=0.75$} at -40 120
\hl{320}{0}{120}
\vl{10}{0}{115}
\vl{10}{320}{115}
\bull{133}{124}
\bull{133}{116}
\bull{160}{116}
\bull{160}{124}
\bull{187}{116}
\bull{187}{124}
\put{$c=1$} at -40 160
\hl{320}{0}{160}
\vl{10}{0}{155}
\vl{10}{320}{155}
\bull{160}{175}
\bull{160}{169}
\bull{160}{163}
\bull{160}{157}
\bull{160}{151}
\bull{160}{145}
\epic
\caption{Minimally dispersed NCNE for different choices of $c$.}\label{leastdispersed}
\end{figure}
 
 \

\textbf{Example.} With six candidates, there are two possible configurations   in an NCNE: we can have three occupied positions with two candidates apiece; or, we can have four occupied positions where the inner two positions are occupied by single candidates. Consider the latter profile first---the former will turn out to be a limiting case of the latter. If (i)-(iii) of Theorem \ref{abb0general} are satisfied, condition (iv) will always be true, since the unpaired candidate at $x^2$ has full-electorate length $\ell(I_2)=I^p+\ell(I_2^R)$, which is clearly larger than all other half-electorates excluding end electorates, which are of length either $I^p$ or $\ell(I_2^R)=\ell(I_3^L)$.  Thus, the only restriction is condition (v).

 To get a maximally dispersed equilibrium, we want $I^p$ to be as small as possible, which means setting $I^p=\ell(I_2^R)=\ell(I_3^L)$. This gives equilibrium profile $x=((x^1,2),(x^2,1), (1-x^2,1),(1-x^1,2))$ where $x^1=\frac{1}{8}(1+3c)$ and $x^2=\frac{3}{8}(1+c)$.   A number of these maximally dispersed equilibria are pictured in Figure \ref{mostdispersed} for a few different values of $c$.

To obtain a minimally dispersed equilibrium, we want $I^p$ to be as large as possible. Condition (iv) will always be satisfied, while condition (v) will still be satisfied if the length of the half-electorates $I_2^R$ and $I_3^L$ go to zero. There, the interior two candidates converge at the median and we are left with minimally dispersed equilibrium profile $x=((x^1, 2), (1/2,2), (1-x^1,2))$ where $x^1=\frac{1}{6}(1+2c)$. Thus, the equilibrium with three occupied positions is the limiting case of the equilibria with four occupied positions. These equilibria are depicted in Figure \ref{leastdispersed}.

\section{Conclusion}
\label{conclusion}

Different voting systems provide political candidates with different incentives and, hence, lead to different outcomes, not all of which are socially desirable. One would usually want a voting system in which adopting extremist positions is not  encouraged while, at the same time, voters are presented with some choice over the policies advocated by the candidates. We have shown that the class of best-worst rules offers a solid middle ground when voters have one positive vote and one negative vote with relatively smaller weight, i.e., so that one negative vote does not cancel out one positive vote. In particular, nonconvergent equilibria exist, and candidates adopt diffferent platforms in a very similar way to under plurality. Importantly, however, the strong best-rewarding incentives of plurality are tempered by the need to avoid negative votes and, indeed, any degree of dispersion between the extreme cases of plurality and full convergence of antiplurality can be obtained for the correct weighting of the negative vote. The need to avoid negative votes leads candidates to moderate their platforms, but without sacrificing diversity entirely.

Though natural, best-worst rules  have not been used in practice, as is the case for many of the voting rules studied in the social choice literature. However, our results provide  evidence that this system is worthy of consideration and presents several desirable properties.  

Future research should investigate the properties of best-worst voting rules in more realistic spatial models with, for example,   probabilistic voting or endogenous candidacy.

\section{Appendix}
\label{appendix}

We include here a number of lemmata that are needed for our main results. Several of these minor results are adapted from results in \citet{cahanslinko}, though similar conditions have appeared in various form in the previous literature since at least \citet{eatonlipsey}. 

The first lemma tells us that the most extreme occupied positions cannot be occupied by single candidates, and they cannot be located at the most extreme points on the issue space.

\begin{lemma}[\citealp{cahanslinko}]
\label{noextremepositions} 
In an NCNE, we must have $n_1,n_q\geq 2$. Moreover, no candidate may adopt the most extreme positions on the issue space. That is, $0<x^1$ and $x^q<1$.
\end{lemma}

\begin{proof}  Evidently,   an unpaired candidate at $x^1$ could move to the right and capture a larger share of positive votes and, at the same time, reduce the number of negative votes. 

To see the second part,  suppose $x^1=0$. Then the at least two candidates at $x^1$ are ranked last equal by a positive measure of voters in the interval $(1-\frac{1}{2}x^q,1]$. By moving to $x^{1+}$, however, a candidate originally at $x^1$ is no longer ranked last by any voters, but still receives the same number of first-place rankings. \end{proof}

The next lemma puts a condition on the continuity  of the function $v_i(t,x_{-i})$ when, in equilibrium, $i$ is at a position occupied by one other candidate and makes a small deviation. 

\begin{lemma}
\label{45cand2} 
Suppose at profile $x$ candidate $i$ is at $x^l$ and $n_l=2$. Then $v_i(x^{l-},x_{-i})+v_i(x^{l+},x_{-i})=2v_i(x)$. In particular, when $x$ is in NCNE, we have $v_i(x^{l-},x_{-i})=v_i(x^{l+},x_{-i})=v_i(x)$.
\end{lemma}

\begin{proof} The issue space can be divided into subintervals of voters who all rank $i$ in the same position.
The immediate interval around $x^l$, $I_l=I_l^L\cup I_l^R$, is the set of voters from which the candidate  receives positive votes.
Let $J$ be the interval of voters from which $i$ receives negative votes. In particular, $J$ is nonempty only if $l=1$ or $l=q$, and it is located at the opposite side of the issue space. Thus, if $l\notin\{1,q\}$, we have 
$v_i(x)=\frac{1}{2}(\ell(I_l^L)+\ell(I_l^R))$. Then $v_i(x^{l-},x_{-i})=\ell(I_l^L)$ and $v_i(x^{l+},x_{-i})=\ell(I_l^R)$. For NCNE, we need  $v_i(x^{l-},x_{-i})\leq v_i(x)$ and $v_i(x^{l+},x_{-i})\leq v_i(x)$. Summing these inequalities, we need $v_i(x^{l-},x_{-i})+v_i(x^{l+},x_{-i})\leq 2v_i(x)$. This, in fact, turns out to be an equality, so that we must have $v_i(x^{l-},x_{-i})=v_i(x^{l+},x_{-i})=v_i(x)$.

If $l=1$ (symmetrically for $l=q$), we have $v_i(x)=\frac{1}{2}(\ell(I_l^L)+\ell(I_l^R))-\frac{c}{2}\ell(J).$  Also, $v_i(x^{l-},x_{-i})=\ell(I_l^L)-c\ell(J)$ and $v_i(x^{l+},x_{-i})=\ell(I_l^R)$. As in the previous case, summing the requirements that these two moves not be beneficial, we find that $v_i(x^{l-},x_{-i})=v_i(x^{l+},x_{-i})=v_i(x)$.\end{proof}

\begin{lemma} 
\label{equalscores} If $n_i=n_j=2$, then $v_i(x)=v_j(x)$ in NCNE. \end{lemma}

\begin{proof} Let $k$ be a candidate at $x^i$ and $l$ be a candidate at $x^j$. 
Note that if $k$ moves to $x^{j+}$ or $x^{j-}$, due to the nature of the voting rule, $k$ receives exactly the same score as $l$ would recieve on moving to $x^{j+}$ or $x^{j-}$. So $v_k(x^{j+},x_{-k})=v_l(x^{j+},x_{-l})$. Then, if $x$ is in NCNE, using Lemma \ref{45cand2} gives that $v_l(x)=v_l(x^{j+},x_{-l})=v_k(x^{j+},x_{-k})\leq v_k(x)$. 
Similarly, $v_l(x^{i+},x_{-l})=v_k(x^{i+},x_{-k})$, from which it follows that $v_k(x)=v_k(x^{i+},x_{-k})=v_l(x^{i+},x_{-l})\leq v_l(x)$. 
So $v_k(x)=v_l(x)$.
\end{proof}

Next, we note that  there cannot be more than two candidates at any position. In particular, this implies that there cannot exist NCNE for $m=3$, a well known result.

\begin{lemma}
\label{nomorethantwo}
In any NCNE,  at any given position there are no more than two candidates. Moreover, $n_1=n_q=2$.
\end{lemma}

\begin{proof}
By Corollary~\ref{BRonly} we have to consider only the case when $c<1$. 

First, we show that,   in NCNE, $n_i\leq 2$ for all $2\leq i \leq q-1$.   If $n_i>2$, where $2\leq i\leq q-1$, then candidate $k$, located at $x^i$ is not ranked last by any voter. Moreover, she is not ranked last by any voter even on deviating to $x^{i+}$ or $x^{i-}$. So the only change in her score on making these moves is from voters in the immediate subintervals $I_1=[(x^{i-1}+x^i)/2,x^{i}]$ and $I_2=[x^{i},(x^{i}+x^{i+1})/2]$, where  voters change candidate $k$ from first equal to first, and from first equal to $n_i$th, respectively. 

In NCNE we must have $$v_k(x^{i-},x_{-k})-v_k(x)= \ell(I_1)- \frac{1}{n_i}(\ell(I_1)+\ell(I_2))\leq 0 $$
and $$v_k(x^{i+},x_{-k})-v_k(x)= \ell(I_2) - \frac{1}{n_i}(\ell(I_1)+\ell(I_2))\leq 0.$$
Adding together these two inequalities we get the requirement that $n_i\leq 2$.

To show that  $n_1=n_q=2$, let us introduce the following notation:
 $I_1=[0,x^1]$, the  voters to the left of candidate 1 (note that by Lemma \ref{noextremepositions}, this set has positive measure); $I_2=[x^1,(x^1+x^2)/2]$, the voters in half the interval between candidates 1 and 2; $I_3= [(x^1+x^q)/2,1]$, the  voters for whom 1 is ranked last equal.

Note that
$ 
v_1(x) =\frac{1}{n_1}(\ell(I_1)+\ell(I_2))-\frac{c}{n_1}\ell(I_3).$
Consider if 1 moves to $x^{1-}$. Then  
$
v_1(x^{1-},x_{-1})= \ell(I_1)-c\ell(I_3)) .
$
If 1 moves to $x^{1+}$ then 
$
v_1(x^{1+},x_{-1})= \ell(I_2).
$
For NCNE we require that these moves not be beneficial to candidate 1.
That is, $v_1(x^{1-},x_{-1})\leq v_1(x)$ which implies we need 
\[
\ell(I_1)-c\ell(I_3))\leq \frac{1}{n_1}(\ell(I_1)+\ell(I_2))-\frac{c}{n_1}\ell(I_3),
 \]
or 
\begin{equation}\label{n1=2.1} \left(1-\frac{1}{n_1}\right)c \ell(I_3)\geq \ell(I_1)-\frac{1}{n_1}(\ell(I_1)+\ell(I_2)).\end{equation}
Similarly, for the other move we have $v_1(x^{1+},x_{-1})\leq v_1(x)$ which gives us
\[
\ell(I_2)\leq \frac{1}{n_1}(\ell(I_1)+\ell(I_2))-\frac{c}{n_1}\ell(I_3),
 \]
implying
$$  \left(1-\frac{1}{n_1}\right)c\ell(I_3)\leq \left(1-\frac{1}{n_1}\right)(\ell(I_1)+\ell(I_2))-(n_1-1)\ell(I_2).$$
 Combining this last equation with \eqref{n1=2.1} yields $(2-n_1)\ell(I_2)\geq 0$, which means $n_1\leq 2$. Hence, $n_1=2$, since we cannot have a lone candidate at $x^1$.
A similar argument gives $n_q=2$.
\end{proof}

\bibliographystyle{chicago}
\bibliography{bestworstbibliography}

\end{document}